\documentclass[letterpaper, 10 pt, conference]{ieeeconf}  % Comment this line out
\IEEEoverridecommandlockouts                              % This command is only
\usepackage{mathptmx} % assumes new font selection scheme installed
\usepackage{times} % assumes new font selection scheme installed
\usepackage{amsmath} % assumes amsmath package installed
\usepackage{amssymb}  % assumes amsmath package installed
\usepackage{mathtools}
\usepackage{./TLC}
\usepackage{physics} % include this after TLC.sty for Im and Re to be redefined
\usepackage{blox}
\usepackage{flushend}

\renewcommand{\bXLinkxy}[3][]{
 \draw [bXLineStyle] (#2) -|
 node[name=#2-#3, above, near start] {#1} (#3) ;
 }

\renewcommand{\bXLinkyx}[3][]{
        \draw [bXLineStyle] (#2.south) |- node[name=#2-#3, above, near end] {#1} (#3) ;
}
\usepackage[style=ieee]{biblatex}
\title{\LARGE \bf
Scaled relative graphs for system analysis
}

\author{Thomas Chaffey$^{1}$ \and Fulvio Forni$^{1}$ \and Rodolphe Sepulchre$^{1}$%
\thanks{The research leading to these results has received funding from the European
Research Council under the Advanced ERC Grant Agreement Switchlet n. 670645.}%
\thanks{$^{1}$The authors are with the University of Cambridge, Department of Engineering, Trumpington Street,
        Cambridge CB2 1PZ, {\tt\small tlc37@cam.ac.uk}, {\tt\small
f.forni@eng.cam.ac.uk},  {\tt\small r.sepulchre@eng.cam.ac.uk}.}
}

\begin{document}

\maketitle
\thispagestyle{empty}
\pagestyle{empty}

%%%%%%%%%%%%%%%%%%%%%%%%%%%%%%%%%%%%%%%%%%%%%%%%%%%%%%%%%%%%%%%%%%%%%%%%%%%%%%%%
\begin{abstract}

        Scaled relative graphs were recently introduced to analyze the convergence of optimization algorithms
        using two dimensional Euclidean geometry.  In this paper, we connect scaled
        relative graphs to the classical theory of input/output systems. It is shown that the
        Nyquist diagram of an LTI system on $L_2$ is the convex hull of its scaled relative
        graph under a particular change of coordinates.  The SRG may be used to
        visualize approximations of static nonlinearities such as the describing
        function and quadratic constraints, allowing system properties to be verified
        or disproved.  Interconnections of systems correspond to graphical
        manipulations of their SRGs.  This is used to provide a simple, graphical
        proof of the classical incremental passivity theorem.
\end{abstract}

\section{Introduction}

The Scaled Relative Graph (SRG) of \textcite{Ryu2021} allows the action of a
nonlinear operator to be visualized on the complex plane.
Incremental properties of an operator, measured between pairs of
inputs and outputs, such as maximal monotonicity and Lipschitz continuity, may
be verified by checking geometric conditions on the SRG of the operator. 
Algebraic manipulations to the operator correspond to geometric manipulations to the
SRG.  This tool allows simple, intuitive and rigorous proofs of the convergence of
many algorithms in convex optimization.  Furthermore, the graphical method is
particularly suitable for proving tightness of convergence bounds, with several novel
tightness results being proved \autocite{Ryu2021, Huang2020}.

In this paper, we connect the scaled relative graph to the classical theory of
input/output systems.  The SRG of a linear, time invariant operator is
the convex hull of its Nyquist diagram under a particular change of coordinates.  The
Nyquist diagram is a cornerstone of linear system theory, and has given rise to many
fundamental developments in the field, among them the Nyquist stability criterion
\autocite{Nyquist1932},
the definition of stability margins and system robustness, the gap metric
\autocite{El-Sakkary1985a} and
graphical interpretation of $H_\infty$ control \autocite{Vinnicombe2000}.  The Nyquist diagram also has a
fundamental place in the theory of nonlinear systems of the Lur'e form (that is,
systems composed of an LTI forward path in feedback with a static nonlinearity).  The
circle and Popov
criteria allow the stability of a Lur'e system to be proved by verifying a
geometric condition on the Nyquist diagram of the LTI component \autocite{Desoer1975}.  The geometric
condition is determined by the properties of the static nonlinearity.   Notably, only
the Nyquist diagram of the LTI component is well-defined, owing to a lack of a
suitable definition of phase for nonlinear systems.  This clearly hampers the use of
the Nyquist diagram in nonlinear system theory.  A notable, classical extension of
the Nyquist diagram to nonlinear systems is the describing function
\autocite{Krylov1947, Blagquiere1966, Slotine1991}.  This
approximate method produces a family of Nyquist curves for a nonlinearity,
parameterized by the amplitude of the input.  
Other efforts to generalize frequency response to nonlinear systems include the work
of \textcite{Pavlov2006} on Bode diagrams for convergent systems, and the recently introduced
notion of nonlinear phase by \textcite{Chen2020}.

Whilst the SRG generalizes the Nyquist diagram of an LTI operator, it may be plotted
for any nonlinear operator, allowing the use of graphical techniques for the analysis
of arbitrary interconnected systems.  We anticipate that, in the same way as the SRG
allows simple, graphical proofs of many results in optimization, the SRG will allow
simple, graphical proofs of the incremental versions of many classical results in 
nonlinear systems and control.  Furthermore, it will allow the popular graphical
control design techniques for LTI systems to be extended to nonlinear systems.

In this preliminary work, we make the first steps towards these aims. We begin in
Section~\ref{sec:def} by defining the SRG over a Hilbert space, and showing how
several important system properties can be determined from the SRG.  
 We then characterize the SRG of an LTI operator on $L_2$ in Section~\ref{sec:LTI}.
 In Section~\ref{sec:NL}, we 
examine the SRGs of static nonlinearities, and show that
standard approximations of static nonlinearities, such as the describing function and
quadratic constraints, have readily computed SRGs which under- and
over-approximate the SRG of the true system.  Finally, in
Section~\ref{sec:interconnection} we demonstrate the usefulness of the SRG in the
analysis of interconnected systems with a simple, graphical proof of the incremental
passivity theorem.

\section{Scaled relative graphs}\label{sec:def}

We define SRGs in the same way as \textcite{Ryu2021}, with the minor modification of
allowing complex valued inner products. Let $\mathcal{H}$ be a Hilbert space,
equipped with an inner product,
$\bra{\cdot}\ket{\cdot}: \mathcal{H} \times \mathcal{H} \to \mathbb{C}$, and the
induced norm $\norm{x} \coloneqq \sqrt{\bra{x}\ket{x}}$.

The angle between $x, y \in \mathcal{H}$ is defined as
\begin{IEEEeqnarray*}{rCl}
        \angle(x, y) \coloneqq \acos \frac{\Re \bra{x}\ket{y}}{\norm{x}\norm{y}}. 
\end{IEEEeqnarray*}

Let $R: \mathcal{H} \to \mathcal{H}$ be an operator, or \emph{system}.  Given $u_1, u_2 \in
\mathcal{U} \subseteq \mathcal{H}$, define the complex number $z_R(u_1, u_2)$ by
\begin{IEEEeqnarray*}{rCl}
        z_R(u_1, u_2) &\coloneqq& \frac{\norm{Ru_1 - Ru_2}}{\norm{u_1 - u_2}}e^{i\angle(u_1 -
        u_2, Ru_1 - Ru_2)}.
\end{IEEEeqnarray*}

The \emph{Scaled Relative Graph} (SRG) of $R$ over $\mathcal{U} \subseteq \mathcal{H}$ is then given by
\begin{IEEEeqnarray*}{rCl}
        \srg[\mathcal{U}]{R} \coloneqq \left\{ z_R(u_1, u_2) \middle| u_1, u_2 \in
        \dom{R}, u_1 \neq u_2 \right\}.
\end{IEEEeqnarray*}
If $\mathcal{U} = \mathcal{H}$, we write $\srg{R} \coloneqq \srg[\mathcal{H}]{R}$.

If $R$ is linear and $\dom{R}$ is a linear subspace of $\mathcal{H}$, $Ru_1 - Ru_2 =
R(u_1 - u_2) = R(v)$ for some $v \in \dom{R}$,
and we can define
\begin{IEEEeqnarray*}{rCl}
        z_R(v) \coloneqq \frac{\norm{Rv}}{\norm{v}}e^{i\angle(v, Rv)}
\end{IEEEeqnarray*}
and
\begin{IEEEeqnarray*}{rCl}
        \srg[\dom{R}]{R} \coloneqq \left\{ z_R(v) \middle| v \in
        \dom{R}, v \neq 0 \right\}.
\end{IEEEeqnarray*}

Each point on the SRG of an operator shows the gain and phase shift
        of a particular pair of inputs.  The SRG thus allows the gain and phase shift
        of nonlinear operators to be visualized graphically, similar to the Nyquist
diagram of an LTI transfer function.

If $\mathcal{A}$ is a class of operators, we define the SRG of $\mathcal{A}$ by
\begin{IEEEeqnarray*}{rCl}
        \srg{\mathcal{A}} \coloneqq \bigcup_{R \in \mathcal{A}} \srg{R}.
\end{IEEEeqnarray*}

A class $\mathcal{A}$, or its SRG, is called \emph{SRG-full} if
\begin{IEEEeqnarray*}{rCl}
        R \in \mathcal{A}\quad \iff \quad \srg{R} \subseteq \srg{\mathcal{A}}.
\end{IEEEeqnarray*}
By construction, the implication $R \in \mathcal{A} \implies \srg{R} \subseteq
\srg{\mathcal{A}}$ is true.  The value of SRG-fullness is in the reverse implication: 
$\srg{R} \subseteq \srg{\mathcal{A}} \implies R \in \mathcal{A}$.  This allows class
membership to be tested graphically.  If $\mathcal{A}$ is the class of systems with a
particular system property, SRG-fullness of $\mathcal{A}$ allows this property to be
verified for a particular operator $R$ by plotting its SRG.  If $\srg{R} \subseteq
\srg{\mathcal{A}}$, $R$ has the property.

We demonstrate this approach
using two classical system properties: $L_2$ gain and incremental positivity (or
monotonicity).  Let $L_2(\mathbb{F})$ be the space of square-integrable signals with values in the
field $\mathbb{F} \in \{\R, \C\}$.  We write $L_2$ when the choice of field is
immaterial.  Define the usual inner product and norm by
\begin{IEEEeqnarray*}{rCl}
        \bra{u}\ket{y} \coloneqq \int_{-\infty}^{\infty} u(t) \bar{y} (t) \dd{t},\\
        \norm{u} \coloneqq \sqrt{\bra{u}\ket{u}},
\end{IEEEeqnarray*}
where $\bar{y}(t)$ denotes the complex conjugate of $y(t)$.

Let $R: L_2 \to L_2$.  $R$ is said to have an \emph{incremental $L_2$ norm bound of
$\gamma$} \autocite{vanderSchaft2017a} if
\begin{IEEEeqnarray*}{rCl}
        \norm{Ru_1 - Ru_2}\leq \gamma\norm{u_1 - u_2}
\end{IEEEeqnarray*}
for all $u_1, u_2 \in L_2$.
$R$ is said to be \emph{incrementally positive}, or \emph{monotone on $L_2$}, if
\begin{IEEEeqnarray*}{rCl}
        \bra{u_1 - u_2}\ket{Ru_1 - Ru_2} \geq 0
\end{IEEEeqnarray*}
for all $u_1, u_2 \in L_2$.  Note that incremental positivity here is meant in the operator
theoretic sense of \autocite[$\sec$ 4 p. 173]{Desoer1975}.  It is closely related to
incremental passivity - indeed, if $R$ is causal, the two are equivalent (the proof
is identical to that of \autocite[Lem. 2, p. 200]{Desoer1975}).  Furthermore, if a
system is linear and time-invariant, incremental passivity is equivalent to
passivity.

The following two propositions demonstrate the verification of system properties from
the system's SRG, and follow directly from \autocite[Prop. 3.3 \& Thm. 3.5]{Ryu2021}.
Both incremental $L_2$ gain and incremental positivity define SRG-full classes.

\begin{proposition}\label{prop:L2}
        An operator $R: L_2 \to L_2$ has an incremental $L_2$ gain less than $\gamma$
        if an only if its SRG lies within the circle centred at the origin of radius
        $\gamma$.
\end{proposition}

This property is reminiscent of the property that the $L_2$ gain of an LTI transfer
function is the maximum magnitude of its frequency response.

\begin{proposition}\label{prop:inc_passive}
        An operator $R: L_2 \to L_2$ is incrementally positive if and only if its SRG
        lies in the right half plane, $\C_{\Re \geq 0}$.
\end{proposition}

This property is reminiscent of the positive realness of a transfer function.

The properties of bounded incremental $L_2$ gain and incremental positivity are particular
examples of incremental Integral Quadratic Constraints (IQCs)
\autocite{Megretski1997}.
A striking corollary of \textcite[Thm. 3.5]{Ryu2021} is that any SRG
defined by a frequency-independent incremental IQC is SRG-full.

\begin{proposition}
        Let $u_i(t)$ denote the input to an arbitrary operator on $L_2$, and $y_i(t)$
        denote the corresponding output.  Let $\Delta u = u_1 - u_2$ and $\Delta y =
        y_1 - y_2$, and $\hat{x}(\omega)$ denote the Fourier transform of signal
        $x(t)$.  Then the classes of operators which obey either of the constraints
        \begin{IEEEeqnarray}{rCl}
                \int^\infty_{-\infty}\begin{pmatrix} \Delta\hat{u}(\omega)\\ \Delta
                \hat{y}(\omega) \end{pmatrix}\tran
                \begin{pmatrix} a & b\\ c & d \end{pmatrix} \begin{pmatrix} \Delta
        \hat{u}(\omega)\\ \Delta \hat{y}(\omega) \end{pmatrix} \dd{\omega} &\geq& 0,\label{eq:quad_1}\\
        \int^\infty_{-\infty} \begin{pmatrix} \Delta u(t)\\ \Delta y(t) \end{pmatrix}\tran
        \begin{pmatrix} a & b\\ c & d \end{pmatrix} \begin{pmatrix} \Delta u(t)\\
\Delta y(t) \end{pmatrix}\dd{t} &\geq& 0,\label{eq:quad_2}
        \end{IEEEeqnarray}
        where $a, b, c, d \in \R$, are SRG-full. 
\end{proposition}

\begin{proof}
       Equation~\eqref{eq:quad_1} gives 
       \begin{IEEEeqnarray*}{rCl}
               a\norm{\Delta \hat{u}}^2 + (b + c)\bra{\Delta \hat{u}}\ket{\Delta
               \hat{y}} + d \norm{\Delta \hat{y}}^2 \geq 0.
       \end{IEEEeqnarray*}
       By Parseval's theorem, this is equivalent to
       \begin{IEEEeqnarray*}{rCl}
               a\norm{\Delta u}^2 + (b + c)\bra{\Delta u}\ket{\Delta y} + d
               \norm{\Delta y}^2 \geq 0,
       \end{IEEEeqnarray*}
       which is also implied by \eqref{eq:quad_2}.  The result then follows from
       \autocite[Thm. 3.5]{Ryu2021}.
\end{proof}

\section{Scaled relative graphs of LTI systems}\label{sec:LTI}

In this section, we investigate the relationship between the SRG of an LTI system and
the classical tool of the Nyquist diagram.
\textcite{Huang2020a} show that the SRG of a diagonal matrix $A$ is the
convex hull of the points $z_A(u_i)$, where $u_i$ are the eigenvectors of $A$, and
the convex hull is taken under the Beltrami-Klein change of coordinates.  Here,
we develop the analogous result for LTI operators on $L_2(\C)$.  We show that the SRG of
such an LTI operator is the convex hull of its Nyquist diagram, under the same
nonlinear change of coordinates.

\subsection{Hyperbolic geometry}\label{sec:hgeo}

We recall some necessary details from hyperbolic geometry.
The notation we use is consistent with \textcite{Huang2020a}.

\begin{definition}\label{def:arc}
        Let $z_1, z_2 \in \C_{\Im \geq 0} \coloneqq \{z \in \C \, | \, \Im(z) \geq 0\}$, the
        upper half complex plane.  We define the following subsets of $\C_{\Im \geq 0}$:
        \begin{enumerate}
                \item 
                        $\text{Circ}\,({z_1},\, {z_2})$ is the circle through $z_1$ and $z_2$ with
                        centre on the real axis.  If $\Re(z_1) = \Re(z_2)$, this is
                        defined as the line extending $[z_1, z_2] \coloneqq \{\alpha
                                z_1 + (1 - \alpha) z_2\, |\, \alpha \in [0, 1]\}$.
                \item $\text{Arc}_{\min}\,(z_1,\, z_2)$ is the arc of
                        $\text{Circ}\,({z_1},\, {z_2})$ in $\C_{\Im \geq 0}$.  If $\Re(z_1) =
                        \Re(z_2)$, then $\text{Arc}_{\min}\,(z_1,\, z_2)$ is $[z_1,
                        z_2]$ (which is the singleton $\{z_1\}$ if $z_1 = z_2$).
                \item Given $z_1, \ldots, z_m \in \C_{\Im \geq 0}$, the \emph{arc-edge polygon}
                        is defined by: $\text{Poly}\,(z_1) \coloneqq \{z_1\}$ and
                        $\text{Poly}\,(z_1, \ldots, z_m)$ is the smallest simply
                        connected set containing $S$, where
                        \begin{IEEEeqnarray*}{rCl}
                                S &=& \bigcup_{i, j = 1\ldots m} \text{Arc}_{\min}\,
                                (z_i, z_j).
                        \end{IEEEeqnarray*}
        \end{enumerate}
\end{definition}
Note that, as $\Poly{z_1, \ldots, z_{m-1}} \subseteq \Poly{z_1, \ldots, z_{m-1},
z_m} \subseteq \C_{\Im \geq 0}$, the set $\Poly{Z}$, where $Z$ is a countably infinite sequence of
points in $\C_{\Im \geq 0}$, is well defined (by the monotone convergence theorem) as the limit $\lim_{m \to \infty} \Poly{Z_m}$, where $Z_m$
is the length $m$ truncation of $Z$.

The notions of Definition~\ref{def:arc} form the basis of the Poincar\'e half plane
model of hyperbolic geometry.  Under the Beltrami-Klein mapping, $\C_{\Im \geq 0}$ is mapped
onto the unit disc, and $\Arc{z_1, z_2}$ is mapped to a straight line segment.  The
Beltrami-Klein mapping is given by $f \circ g$, where
\begin{IEEEeqnarray*}{rCl}
        f(z) &=& \frac{2z}{1 + |z|^2},\\
        g(z) &=& \frac{z - i}{z + i}.
\end{IEEEeqnarray*}
We make the following definitions of convexity and the convex hull in the Poincar\'e half plane model.

\begin{definition}
        A set $S \subseteq \C_{\Im \geq 0}$ is called \emph{hyperbolic-convex} or \emph{h-convex} if 
        \begin{IEEEeqnarray*}{rCl}
                z_1, z_2 \in S \implies \Arc{z_1, z_2} \in S.
        \end{IEEEeqnarray*}
        Given a set of points $P \in \C_{\Im \geq 0}$, the \emph{h-convex hull of} $P$ is the
        smallest h-convex set containing $P$.
\end{definition}

Note that h-convexity is equivalent to Euclidean convexity under the Beltrami-Klein
mapping.  $\Arc{z_1, z_2}$ is the minimal geodesic between $z_1$ and $z_2$ under
the Poincar\'e metric, so h-convexity may be thought of as geodesic convexity with
respect to this metric.  

\subsection{LTI SRGs and the Nyquist diagram}

Let $g: L_2(\C) \to L_2(\C)$ be linear and time invariant, and denote its transfer
function by $G(s)$.
$g$ maps a complex sinusoid $u(t) = a e^{j\omega t}$ to the complex sinusoid\footnote{While these complex exponentials do not belong to $L_2(\C)$,
they can be treated as the limit of a series of signals in $L_2(\C)$ (for example,
truncations to a finite number of periods), and their inner products can be computed
accordingly.}
$y(t) = |G(j \omega)| e^{\angle G( j \omega) + j \omega t}$.  

\begin{definition}
        The Nyquist diagram $\nyq{G}$ of an operator $g: L_2(\C) \to L_2(\C)$ is the locus of points
        $\{G(j \omega)\,|\, \omega \in \R\}$.
\end{definition}

\begin{theorem}\label{prop:nyq}
        Let $g: L_2(\C) \to L_2(\C)$ be linear and time invariant.
        $\srg{G} \cap \C_{\Im \geq 0}$ is the h-convex hull of $\nyq{G} \cap \C_{\Im \geq 0}$.
\end{theorem}

The proof of Theorem~\ref{prop:nyq} is closely related to the proof of \textcite[Thm.
3.1]{Huang2020a}, and may be found in the submitted journal version of this paper \autocite{Chaffey2021c}.  A consequence of
Theorem~\ref{prop:nyq} is that the SRG of an LTI operator is bounded by its
Nyquist diagram.

Given Theorem~\ref{prop:nyq}, we recover the following two familiar properties of
the Nyquist diagram as special cases of Propositions~\ref{prop:L2}
and~\ref{prop:inc_passive}.

\begin{corollary}
        The $L_2$ gain of a stable transfer function $G(s)$ is the largest magnitude of its
        Nyquist diagram, $\max_{\omega \in \R} |G(j\omega)|$.
\end{corollary}

\begin{corollary}
        A causal transfer function $G(s)$ is passive if and only if its Nyquist diagram lies
        in the right half plane.
\end{corollary}

        The Nyquist diagram of the first order lag
        \begin{IEEEeqnarray*}{rCl}
                G(s) = \frac{1}{s + 1}
        \end{IEEEeqnarray*}
        is the circle in $\C$ with centre $0.5$ and radius $0.5$
        (Figure~\ref{fig:lag}, top).  Under the
        Beltrami-Klein transformation, this is a straight line (which is evident as
        it is a circle with centre on the real axis), and is therefore its own
        h-convex hull (Figure~\ref{fig:lag}, bottom).  It follows that $\srg{G} = \nyq{G}$.

        \begin{figure}[h]
                \centering
                \includegraphics[width=0.8\linewidth]{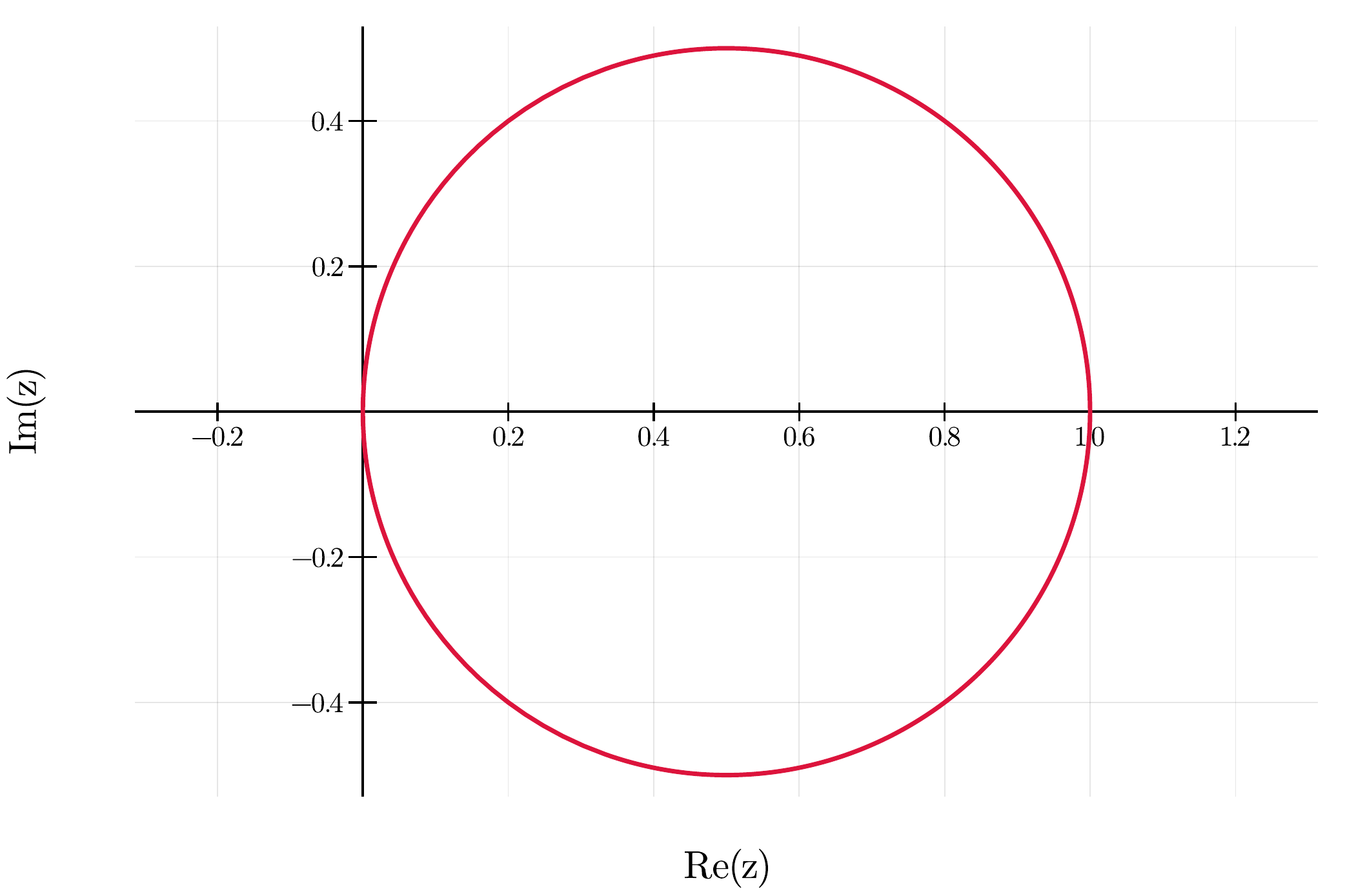}
                \includegraphics[width=0.8\linewidth]{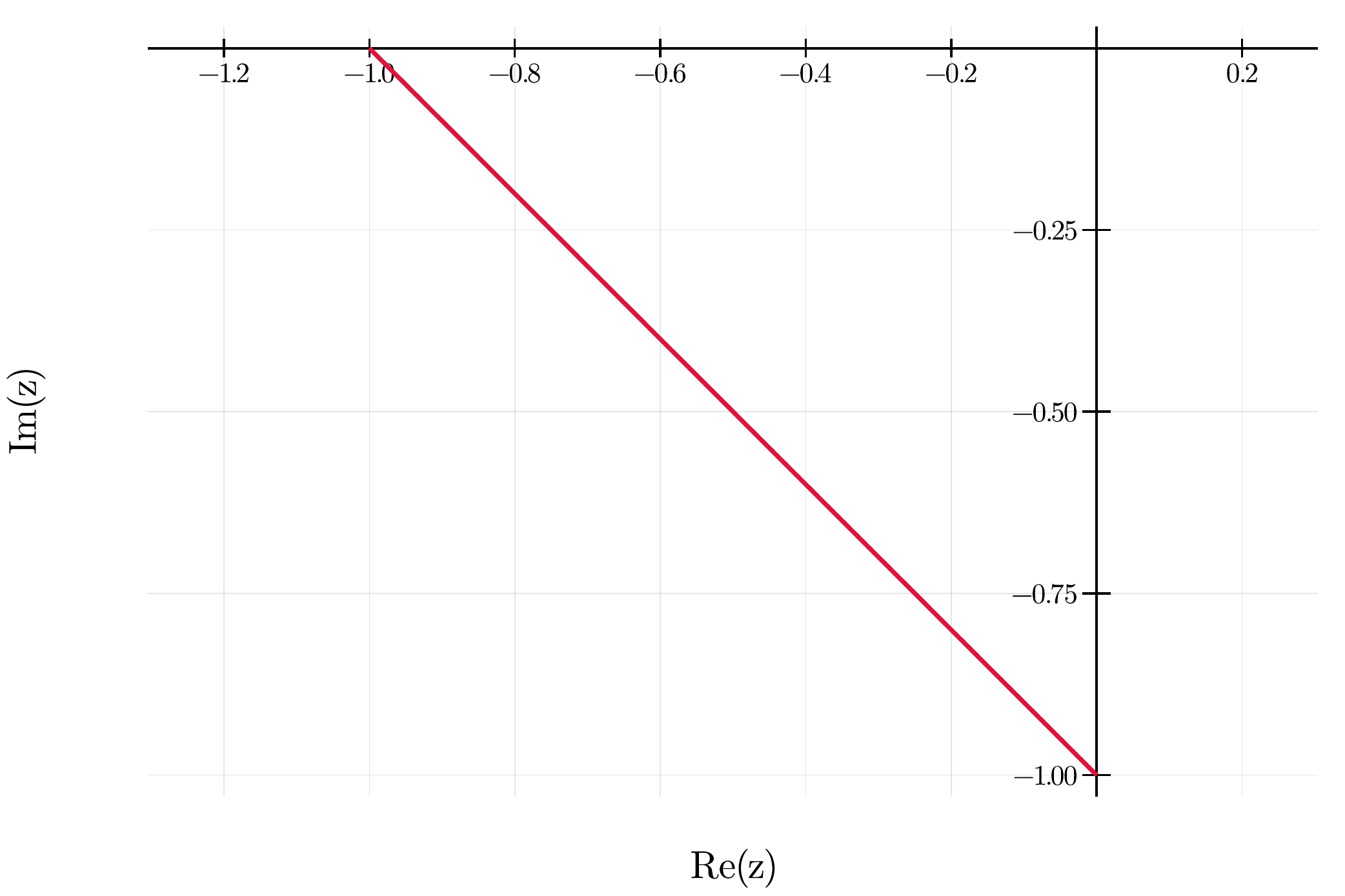}
                \caption{The SRG of the first order lag $1/(s + 1)$ is its Nyquist
                diagram (top).  Under the Beltrami-Klein mapping, the upper half of
                the Nyquist diagram maps to a straight line on the unit disc (bottom).}%
                \label{fig:lag}
        \end{figure}

        Two examples of transfer functions whose Nyquist diagrams are proper subsets
        of their SRGs are illustrated in Figures~\ref{fig:second_order} and
        \ref{fig:third_order}.  These are the systems $1/(s^2 + 2s + 1)$ and $1/(s^3
        + 5s^2 + 2s + 1)$ respectively.  The upper plots illustrate the SRGs, while
        the lower plots show the corresponding regions under the Beltrami-Klein
        mapping.

        \begin{figure}[h]
                \centering
                \includegraphics[width=0.8\linewidth]{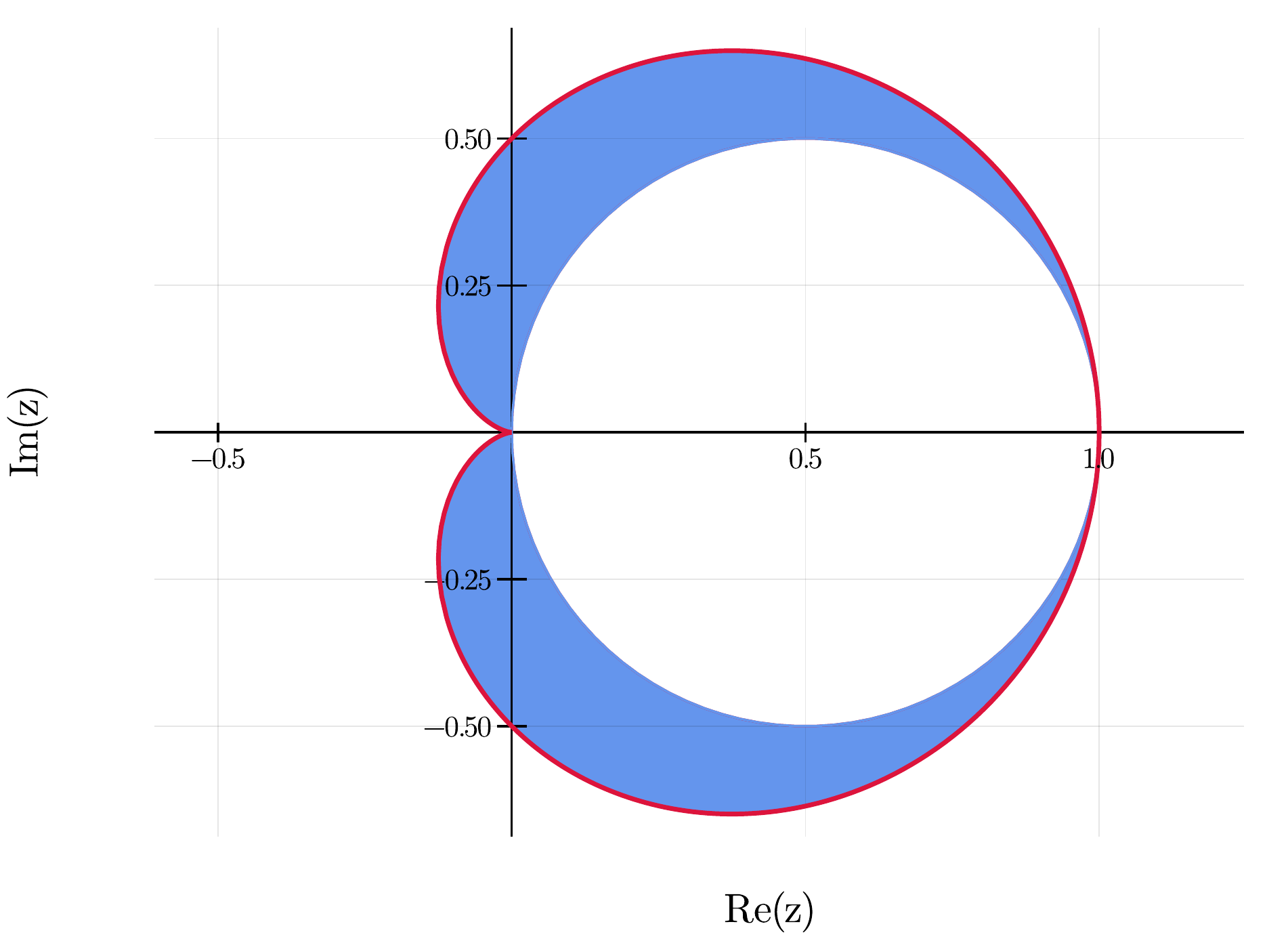}
                \includegraphics[width=0.8\linewidth]{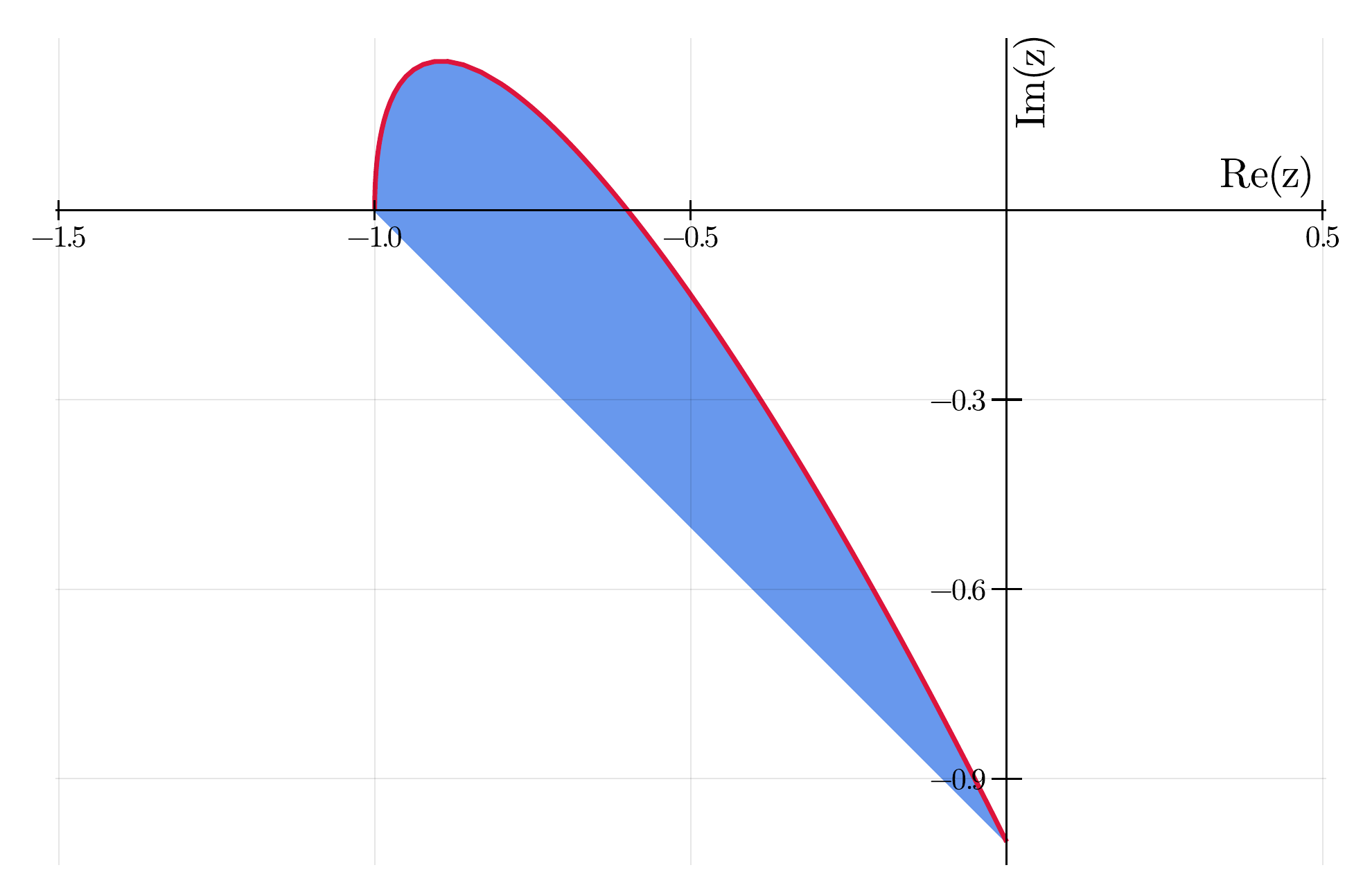}
                \caption{SRG of the transfer function $1/(s^2 + 2s + 1)$ (top). The red
                        curve is its Nyquist diagram.  The image of the upper half of
                        the SRG under the Beltrami-Klein mapping is shown below.}%
                \label{fig:second_order}
        \end{figure}

        \begin{figure}[h]
                \centering
                \includegraphics[width=0.8\linewidth]{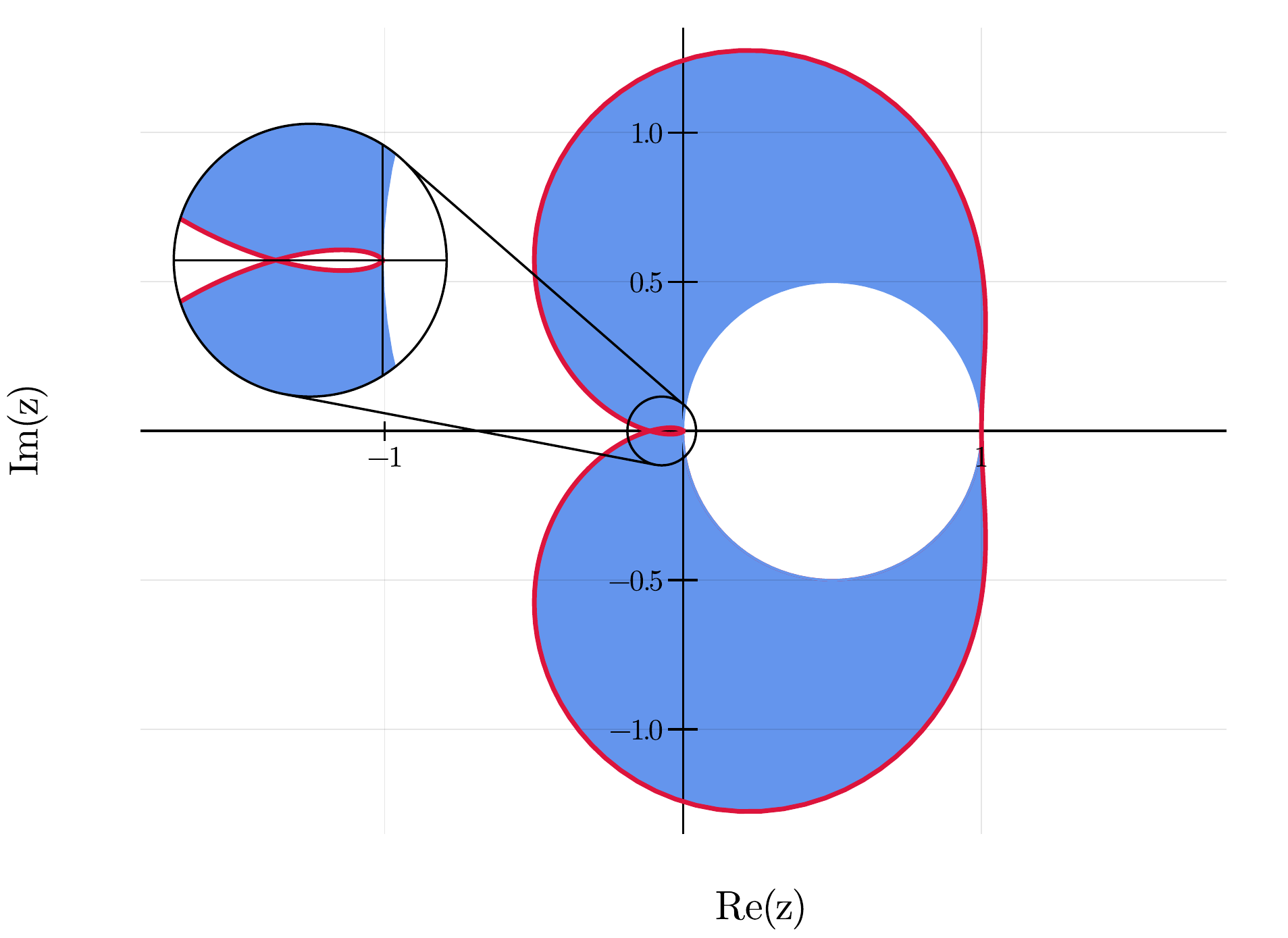}
                \includegraphics[width=0.8\linewidth]{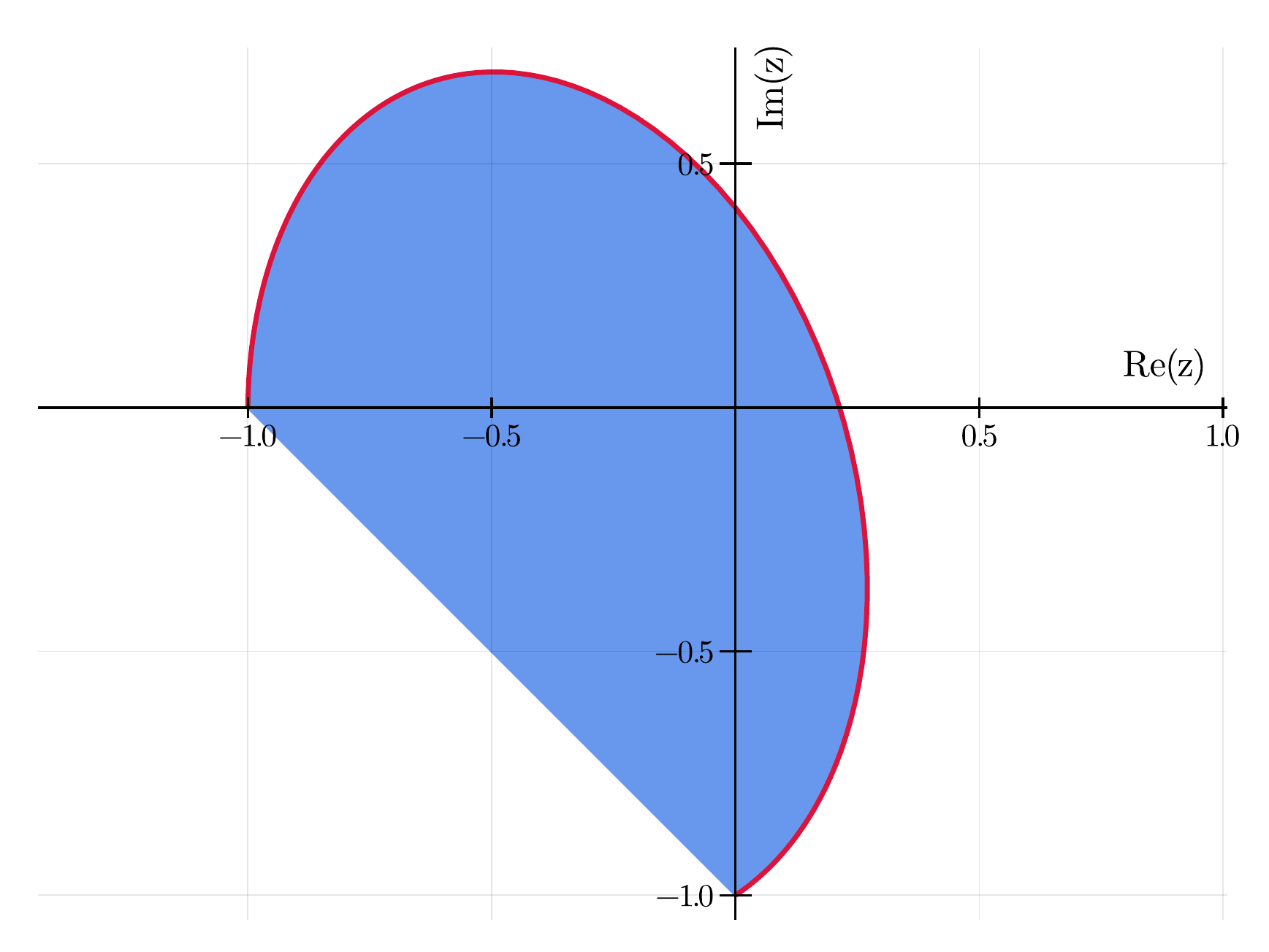}
                \caption{SRG of the transfer function $1/(s^3 + 5s^2 + 2s + 1)$ (top). The
                        red curve is its Nyquist diagram. The image of the upper half of
                        the SRG under the Beltrami-Klein mapping is shown below.}%
                \label{fig:third_order}
        \end{figure}

\section{Scaled relative graphs of static nonlinearities}\label{sec:NL}

The second class of systems we consider are static nonlinearities, that is, systems
governed by a relation
\begin{IEEEeqnarray*}{rCl}
        y(t) = \phi(u(t))
\end{IEEEeqnarray*}
between the input at time $t$ and the output at time $t$.  In general, it is difficult to 
analytically compute the SRG of a static nonlinearity.
However, valuable information about a system can be
obtained from under- and over-approximations of its SRG.  In this section, we
describe several methods for approximating the SRG of a static nonlinearity,
demonstrated on the running example of a saturation, $y = \sat{u}$:
\begin{IEEEeqnarray*}{rCl}
        \sat{u} \coloneqq \begin{cases}{}
                        -1&u < -1\\
                        u &|u| \leq 1\\
                        1 &u > 1.
                \end{cases}
\end{IEEEeqnarray*}

\subsection{Sampling the scaled relative graph}
The simplest method of approximating the SRG of a system is to sample the input space
and directly compute the SRG over these samples, producing a subset of the full SRG.

Fourier analysis allows the computation of the SRG over individual samples to be made
more computationally tractable.  While computing the inner product over a continuous
signal may in general be expensive, by Parseval's theorem,
we have $\bra{u}\ket{y} = \bra{\hat{u}}\ket{\hat{y}}$, where $\hat{x}$ represents the
Fourier transform of $x$.  If the input signals are
chosen to have a small number of nonzero Fourier coefficients (like, for example,
sinusoids), inner products involving the inputs $u_1, u_2$ can be computed using a small number of arithmetic
operations.  However, for a nonlinear system, the spectrum of the output may still be
infinite.   The action of the nonlinear system can be approximated by restricting to
a finite number of harmonics of the output.

The classical method of describing function analysis approximates the response of a
nonlinear system, excited by a sinusoidal input, by the fundamental frequency
component of the output \autocite{Krylov1947}. This provides a generalization of a
transfer function which is, in general,
dependent on both the amplitude and frequency of the input sinusoid.  If higher
harmonics are filtered out by other components of the system, this provides a
reasonable approximation of the nonlinear operator \autocite{Blagquiere1966,
Slotine1991}.

Let $R: L_2(\C) \to L_2(\C)$ be a nonlinear operator, $u(t) = ae^{j\omega t}$ and
$y(t) = Ru(t)$.  If $R$ is a static nonlinearity, $y$ will be periodic with the
same period as $u$, and can be expanded in 
the Fourier series $y(t) = \sum_{n=0}^\infty \hat{y}(n) e^{jn\omega t}$.
\emph{The describing function of $R$} is defined as $\df{R}(a, \omega) =
\hat{y}(1)/a$.  The describing function defines an operator $\df{R}$ on the set of complex
exponentials by $ae^{j\omega t} \mapsto \df{R}(a, \omega) ae^{j\omega t}$.  This
operator can be visualized on an SRG by restricting inputs $u_1, u_2$ to be of the
form $u_i(t) = a_i e^{j \omega_i t + j \psi_i}$.

For several common nonlinearities of practical importance, the describing function
can be computed analytically - see, for example, \autocite{Slotine1991}.
The describing function of the saturation is given by
\begin{IEEEeqnarray*}{rCl}
        \df{\sat{u}}(a) = \left\{ \begin{array}{c c}
                        1&\quad |a| < 1\\
                        \frac{2}{\pi}\left(\asin \frac{1}{a} + \frac{1}{a}\sqrt{1 -
        \frac{1}{a^2}}\right) &\quad |a| > 1.
                \end{array}\right.
\end{IEEEeqnarray*}
The SRG of this describing function is illustrated in Figure~\ref{fig:sat_df}.  This
SRG is computed over pairs of inputs $u_1(t) = a_1 \sin(\omega t)$, $u_2(t) = a_2 \sin(\omega t)$, 
where $\omega$ is arbitrary and $a_1, a_2$ are variables in $\R$.
\begin{figure}[h]
        \centering
        \includegraphics[width=0.8\linewidth]{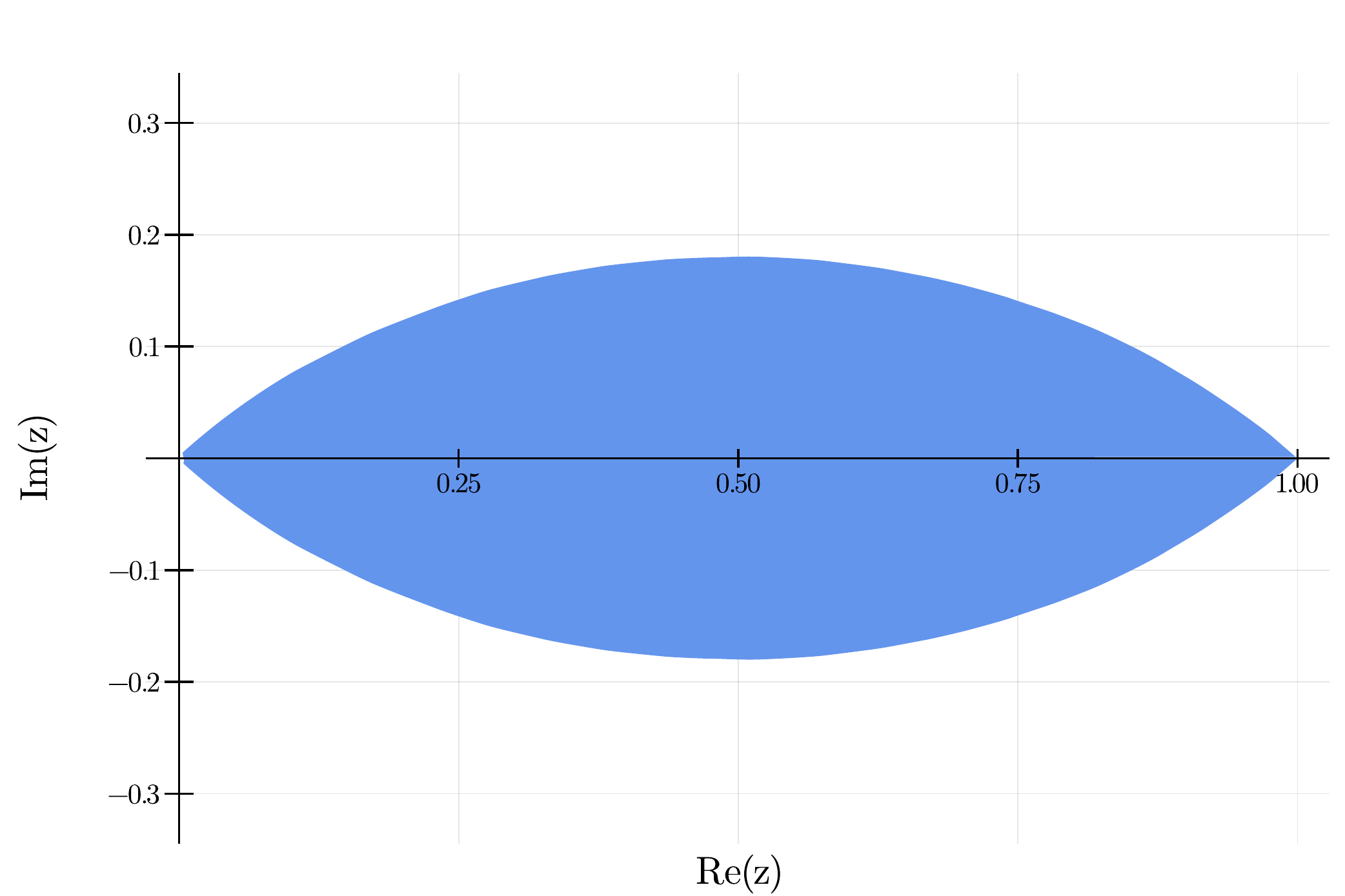}
        \caption{SRG of the describing function of a saturation over the inputs
                $u_1(t)
        = a_1 \sin(\omega t)$, $u_2(t) = a_2 \sin(\omega t)$.}%
        \label{fig:sat_df}
\end{figure}

Truncating the spectrum of the output of a nonlinearity to the fundamental frequency
component leads to an underestimation of the action of the nonlinearity, in the
following sense.

\begin{proposition}
        Given an operator $R: L_2(\C) \to L_2(\C)$ and inputs $u_1(t) = a_1 e^{j
        \omega t}$, $u_2(t) = a_2 e^{j \omega t}$, we have
        \begin{IEEEeqnarray*}{rCl}
                |z_{\df{R}}(u_1, u_2)| \leq |z_{R}(u_1 - u_2)|,&&\\
                \angle((u_1 - u_2),\, (\df{R}(u_1)-\df{R}(u_2))) &\leq&\\
                \angle((u_1 - u_2),\, (Ru_1-Ru_2)).
        \end{IEEEeqnarray*}
\end{proposition}

\begin{proof}
        From Bessel's inequality, we have
        \begin{IEEEeqnarray}{rCl}
                \norm{Ru_1 - Ru_2} \geq \norm{\hat{y}_1(1) -
                \hat{y}_2(1)},\label{eq:bessel}
        \end{IEEEeqnarray}
        where $\hat{y}_1(1), \hat{y}_2(1)$ are the first complex Fourier coefficients of
        $Ru_1, Ru_2$.  These are the describing function approximation of
        $R$ applied to $u_1$ and $u_2$.  It follows that the gain of the describing
        function approximation is at most the gain of the original function.

        Since $u_1 - u_2 = a_0e^{j\omega t}$, it follows from Parseval's theorem
        and the orthogonality of $\{e^{jn\omega t}\}_{n \in \mathbb{N}}$ %N or Z
        that $\bra{u_1 - u_2}\ket{Ru_1 - Ru_2} = \bra{u_1 - u_2}\ket{\df{R}(u_1) -
        \df{R}(u_2)}$.  The second inequality then follows from \eqref{eq:bessel} and the
        fact that $\acos$ is monotonically decreasing.
\end{proof}

A better approximation of the SRG of a nonlinearity can be gained by considering
more terms of the Fourier series of the output - this is the technique adopted
in higher-order extensions of the describing function \autocite{Nuij2006}.
We can also consider a larger class of input signals.
Figure~\ref{fig:harmonic_sat} shows
the SRG of a saturation, computed over the inputs $u_1 =k_1 + a_1 \sin(\omega t)$, $u_2 =
k_2 + a_2 \sin(\omega t)$, with the first 10 harmonics of the output calculated.

\begin{figure}[h]
        \centering
        \includegraphics[width=0.8\linewidth]{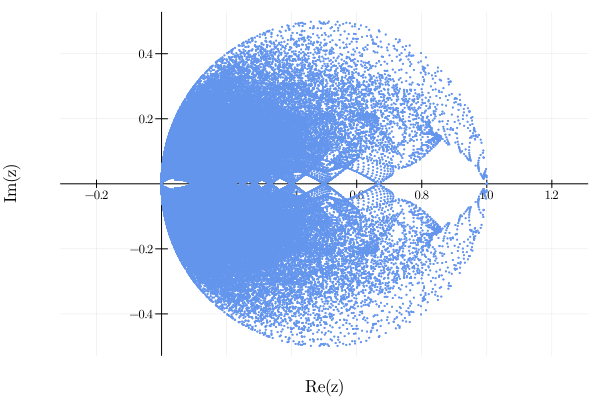}
        \caption{The SRG of a saturation, computed over a set of inputs of the form $u_1 = k_1 + a_1 \sin(\omega t)$, $u_2 =
k_2 + a_2 \sin(\omega t)$, with the output truncated in frequency to the first 10
harmonics.}%
        \label{fig:harmonic_sat}
\end{figure}

For differentiable nonlinearities, taking the discretization of the nonlinearity
along an input trajectory (that is, the Fr\'echet derivative of the nonlinearity)
gives an approximation of the system behavior at nearby trajectories.  Plotting the
SRG of the linearization provides another method of approximating the SRG of the
original system.  For the saturation, taking the linearization along constant
trajectories $u(t) \neq -1, 1$ gives either the identity or the zero map - the SRG of
the linearization is then the point $0$ if $|u(t)| < 1$ or the point $1$ otherwise.

\subsection{Bounding the scaled relative graph}

The SRG of a system is constrained by its input/output properties.  By plotting the
SRGs of each system property, we can build a set of geometric constraints on the SRG
of the system, using the fact that, for two classes of operators $\mathcal{A}$ and
$\mathcal{B}$, $\srg{\mathcal{A} \cap \mathcal{B}} \subseteq \srg{\mathcal{A}}\cap\srg{\mathcal{B}}$.  Theorem 4.1 of
\textcite{Ryu2021} shows that for SRG-full classes, a stronger result
holds.
\begin{proposition}\label{prop:intersect}
        (Theorem 4.1 \autocite{Ryu2021}): If $\mathcal{A}$ and $\mathcal{B}$ are
        SRG-full classes, then $\mathcal{A} \cap \mathcal{B}$ is SRG-full, and
        \begin{IEEEeqnarray*}{rCl}
                \srg{\mathcal{A} \cap \mathcal{B}} = \srg{\mathcal{A}}\cap
                \srg{\mathcal{B}}.
        \end{IEEEeqnarray*}
\end{proposition}

We demonstrate over-approximation of the SRG again using the example of a saturation.
The saturation obeys the following two slope, or incremental sector, conditions.
\begin{IEEEeqnarray*}{rCl}
        \bra{u_1 - u_2}\ket{\sat{u_1} - \sat{u_2}} &\geq& 0\\
        \norm{\sat{u_1} - \sat{u_2}} &\leq& \norm{u_1 - u_2}.
\end{IEEEeqnarray*}

Note that the two conditions above are equivalent to the standard IQC for
incrementally sector-bounded nonlinearities. Denoting $u_1(t) - u_2(t)$ by $\Delta
u$ and $y_1(t) - y_2(t)$ by $\Delta y$, we have \autocite[Thm. 2, p.2]{Desoer1975}:
\begin{IEEEeqnarray*}{lrCl}
       & \frac{\Delta y}{\Delta u} \leq 1 \;\&\; \Delta y \Delta u &\geq& 0\\
        \iff & \Delta y\Delta u \leq \Delta u^2 \;\&\; \Delta y\Delta u &\geq& 0\\
        \iff & (\Delta y \Delta u - \Delta u^2)\Delta y \Delta u &\leq& 0\\
        \iff & \Delta y^2 - \Delta u \Delta y &\leq& 0\IEEEyesnumber\label{eq:osp}\\
        \iff &  \frac{1}{2}\begin{pmatrix} \Delta u\\ \Delta y \end{pmatrix}\tran
             \begin{pmatrix} 0 & 1\\ 1 & -2 \end{pmatrix} \begin{pmatrix} \Delta u\\
\Delta y \end{pmatrix} &\geq& 0.
\end{IEEEeqnarray*}

Integrating and applying Parseval's theorem, we have the incremental IQC
\begin{IEEEeqnarray*}{rCl}
\frac{1}{2}\int_{-\infty}^\infty \begin{pmatrix} \Delta \hat{u}(\omega)\\ \Delta
        \hat{y}(\omega) \end{pmatrix}\tran
        \begin{pmatrix} 0 & 1\\ 1 & -2 \end{pmatrix} \begin{pmatrix} \Delta
        \hat{u}(\omega)\\
\Delta \hat{y}(\omega) \end{pmatrix}\dd{\omega} \geq 0,
\end{IEEEeqnarray*}
where $\hat{x}(\omega)$ denotes the Fourier transform of $x(t)$.

Equation~\ref{eq:osp} states that the saturation is $1$-cocoercive.  It follows from
\autocite[Prop. 3.3]{Ryu2021} that the SRG of the saturation is contained in the disc
with centre $1/2$ and radius $1/2$, as shown in Figure~\ref{fig:sector_SRG}.

\begin{figure}[h]
        \centering
        \includegraphics[width=0.8\linewidth]{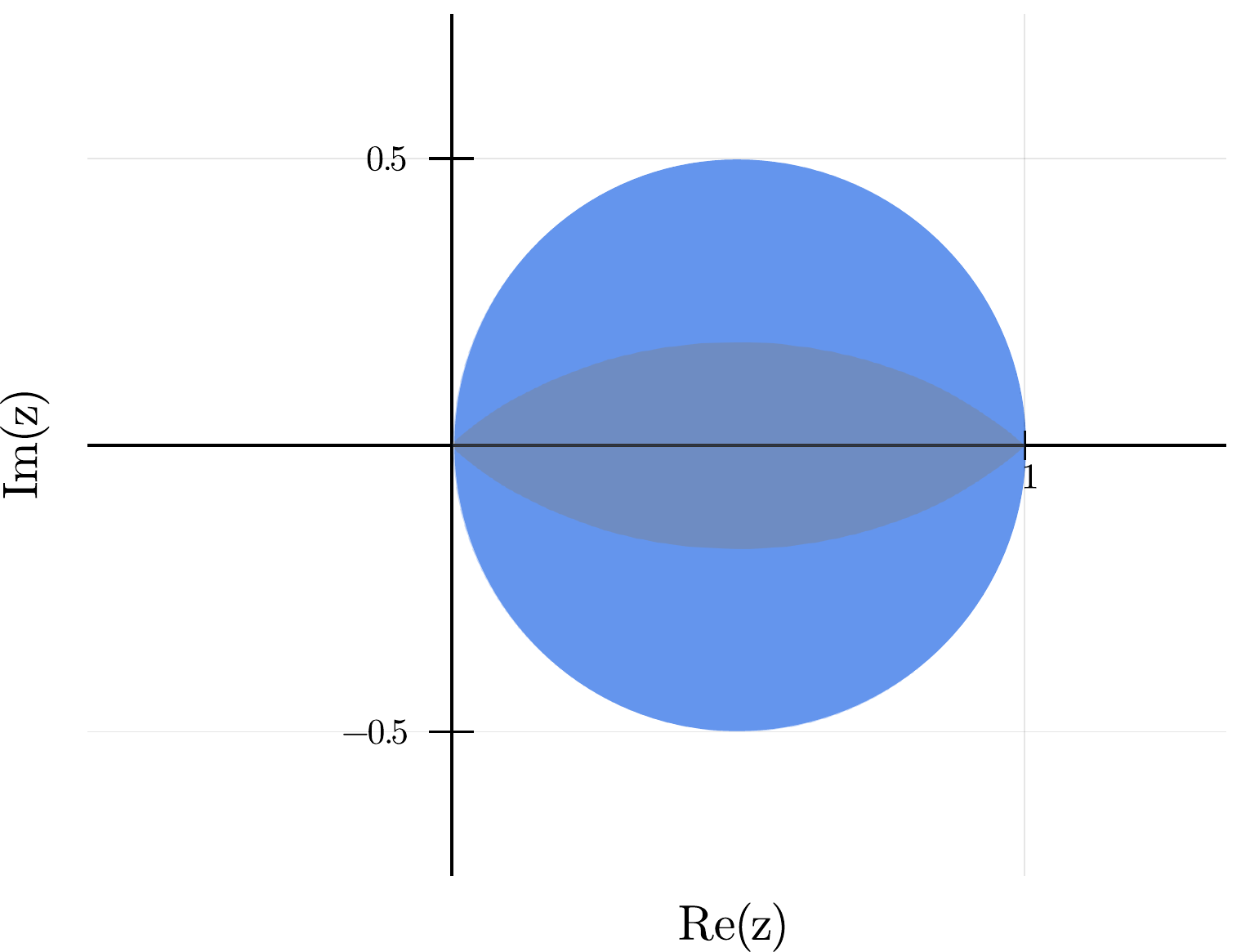}
        \caption{The SRG of the class of operators which are incrementally positive
        and have incremental $L_2$ gain of at most 1. The shadow indicates the
SRG of the describing function approximation.}%
        \label{fig:sector_SRG}
\end{figure}

The SRG shown in Figure~\ref{fig:sector_SRG} is SRG-full: if the SRG of a system $H$
lies within this SRG, then $H$ is incrementally passive and has $L_2$ gain less than
or equal to 1.  

In the previous section, we showed that sampling the SRG of the saturation appeared
to fill the disc with centre $1/2$ and radius $1/2$.  Here, we have shown that this
region bounds the SRG of the saturation.  In the submitted journal version of this
paper \autocite{Chaffey2021c}, we have refined this analysis, and shown that the SRG
of the saturation is \emph{precisely} the disc with centre $1/2$ and radius $1/2$.

\section{System analysis with scaled relative graphs}\label{sec:interconnection}

Under certain conditions, the SRG of an interconnection of systems is the
interconnection of the SRGs of the individual systems.  In the same way as the SRG can be used
to give intuitive, visual proofs of many standard results in optimization, it can be
used to give simple proofs of many classical results in system theory, which rely on
the inference of the properties of a system from the properties of its components.  We
demonstrate this with a simple, graphical proof of the classical incremental passivity theorem.

First, we introduce the necessary interconnection laws (and refer the reader to
\autocite[$\sec$ 4]{Ryu2021} for several other interconnection laws).  Given an
operator $R$, we denote by $R^{-1}$ the \emph{relational inverse} of $R$, that is the
map $u \mapsto \{v\, |\, Rv = u\}$.  This map always exists, and coincides with the
regular inverse when $R$ is an invertible operator.  We define inversion in the
complex plane by the M\"obius transformation $re^{j\omega} \mapsto (1/r)e^{j\omega}$.

\begin{proposition}\label{prop:inversion}
        (Theorem 4.3 \autocite{Ryu2021}): Given a class of operators $\mathcal{A}$, 
        $\srg{\mathcal{A}^{-1}} = (\srg{\mathcal{A}})^{-1}$. 
\end{proposition}

A class of operators $\mathcal{A}$ is said to
satisfy the chord property if $z \in \srg{\mathcal{A}}\setminus\{\infty\}$ implies $[z,
\bar z] \subseteq \srg{\mathcal{A}}$.
\begin{proposition}\label{prop:summation}
        (Theorem 4.4 \autocite{Ryu2021}): Let $\mathcal{A}$ and $\mathcal{B}$ be
        SRG-full classes such that either $\mathcal{A}$ or $\mathcal{B}$ satisfies
        the chord property and $\infty \nin \srg{\mathcal{A}}$ and $\infty \nin
        \srg{\mathcal{B}}$.  Then $\srg{\mathcal{A} + \mathcal{B}} = \srg{\mathcal{A}} +
        \srg{\mathcal{B}}$.
\end{proposition}

Propositions~\ref{prop:inversion} and \ref{prop:summation} allow a simple, geometric
proof of the following incremental form of the classical passivity theorem.

\begin{proposition}
        The negative feedback interconnection (Figure~\ref{fig:feedback}) of two incrementally positive systems $H_1$ and
        $H_2$, is incrementally positive.
\end{proposition}

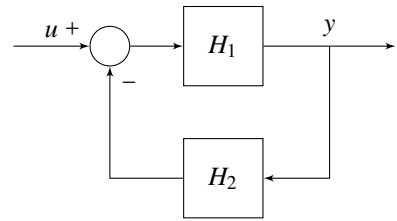
\begin{figure}[h]
        \centering
        \begin{tikzpicture}
                \bXInput{input}
                \bXComp*{sum}{input}
                \bXLink[$u$]{input}{sum}
                \bXBloc[2]{one}{$H_1$}{sum}
                \bXLink{sum}{one}
                \bXOutput[5]{out}{one}
                \bXLink[$y$]{one}{out}
                \bXBranchy{out}{return}
                \bXBlocr[5]{two}{$H_2$}{return}
                \bXLinkyx{one-out}{two}
                \bXLinkxy{two}{sum}
        \end{tikzpicture}
        \caption{Negative feedback interconnection of $H_1$ and $H_2$.}%
        \label{fig:feedback}
\end{figure}
\begin{proof}
        The negative feedback interconnection of $H_1$ and $H_2$ may be written as
\begin{IEEEeqnarray*}{rCl}
        y \in (H_1^{-1} + H_2)^{-1}(u).
\end{IEEEeqnarray*}
The proof by SRG shows that the sequence of operations that take $H_1$ to this
form leave the right half plane $\mathbb{C}_{\Re \geq 0}$
invariant.  In particular, note that SRG($\mathcal{M}$) $= \mathbb{C}_{\Re \geq 0}$, and
$\text{SRG}(H_1) \subseteq \text{SRG}(\mathcal{M})$.  Furthermore, $\mathcal{M}$ is
SRG-full and obeys the chord condition. Now,
\begin{enumerate}
        \item $\text{SRG}(\mathcal{M}^{-1}) = (\mathbb{C}_{\Re \geq 0})^{-1} = \mathbb{C}_{\Re \geq 0}$;
        \item $\text{SRG}(\mathcal{M}^{-1} + \mathcal{M}) = \mathbb{C}_{\Re \geq 0} +
                \mathbb{C}_{\Re \geq 0} = \mathbb{C}_{\Re \geq 0}$;
        \item $\text{SRG}(\mathcal{M}^{-1} + \mathcal{M})^{-1} = (\mathbb{C}_{\Re
                \geq 0})^{-1} = \mathbb{C}_{\Re \geq 0}$.\qedhere
\end{enumerate}
\end{proof}

\section{Conclusions}

This paper has presented preliminary results applying the scaled relative graph
of \textcite{Ryu2021} to system analysis.  The SRG is a generalization of the
classical Nyquist criterion which may be plotted for any operator on $L_2$, not only
those that are linear and time invariant.  This opens many opportunities to revisit
classical analysis and control design techniques in terms of the SRG, and extend
linear techniques to nonlinear operators.  We have made a preliminary step in this
direction by presenting a simple, graphical proof of the passivity theorem.

A particularly promising avenue is the
extension of LTI stability and robustness techniques to nonlinear operators via the SRG, beginning
with a generalization of the Nyquist criterion for stable operators. The SRG
allows the gain and phase margins to be calculated for nonlinear
operators. This approach to nonlinear system analysis is explored in the submitted journal
version of this paper \autocite{Chaffey2021c}.

\printbibliography

\end{document}